\documentclass[12pt,a4paper]{article}
\usepackage{amssymb,amsmath,amsthm,eucal,cite,subfig,graphicx}
\usepackage[top=2.5cm,right=2.5cm,left=2.5cm,bottom=2.5cm]{geometry}
\usepackage{titling} 

\allowdisplaybreaks[4]

\title{Metric Properties of the Fuzzy Sphere} 

\author{Francesco D'Andrea,$^1$ Fedele Lizzi$\hspace{1pt}^{2,3,4}$
and Joseph C.~V\'arilly$^{\,5}$ \\[12pt]
{\footnotesize $^1$ Dipartimento di Matematica e Applicazioni,
Universit\`a di Napoli ``Federico II'', Italy}\\[3pt]
{\footnotesize $^2$ Dipartimento di Scienze Fisiche,
Universit\`a di Napoli ``Federico II'', Italy
}\\[3pt]
{\footnotesize $^3$ INFN Sezione di Napoli, Italy
}\\[3pt]
{\footnotesize $^4$ Departament de Estructura i Constituents
de la Mat\`eria,}\\
{\footnotesize and Institut de Ci\'encies del Cosmos,
Universitat de Barcelona, Barcelona, Catalonia, Spain}
\\[3pt]
{\footnotesize $^5$ Escuela de Matem\'atica,
Universidad de Costa Rica, San Jos\'e 11501, Costa Rica}
}

\date{} 

\DeclareMathOperator{\ad}{ad}     
\DeclareMathOperator{\End}{End}   
\DeclareMathOperator{\sign}{sign} 
\DeclareMathOperator{\Tr}{Tr}     
\DeclareMathOperator{\tsum}{{\textstyle\sum}} 

\newcommand{\al}{\alpha}          
\newcommand{\Dl}{\Delta}          
\newcommand{\eps}{\varepsilon}    
\newcommand{\Ga}{\Gamma}          
\newcommand{\ga}{\gamma}          
\newcommand{\la}{\lambda}         
\newcommand{\om}{\omega}          
\newcommand{\Sg}{\Sigma}          
\newcommand{\sg}{\sigma}          
\renewcommand{\th}{\theta}        
\newcommand{\vf}{\varphi}         

\newcommand{\sA}{\mathcal{A}}     
\newcommand{\sB}{\mathcal{B}}     
\newcommand{\sD}{\mathcal{D}}     
\newcommand{\sH}{\mathcal{H}}     
\newcommand{\sJ}{\mathcal{J}}     
\newcommand{\sL}{\mathcal{L}}     
\newcommand{\sS}{\mathcal{S}}     
\newcommand{\sU}{\mathcal{U}}     

\newcommand{\bC}{\mathbb{C}}      
\newcommand{\bN}{\mathbb{N}}      
\newcommand{\bR}{\mathbb{R}}      
\newcommand{\bS}{\mathbb{S}}      
\newcommand{\bZ}{\mathbb{Z}}      

\newcommand{\Cl}{\mathrm{C}\ell}  
\newcommand{\geo}{\mathrm{geo}}   
\newcommand{\ii}{\mathrm{i}}      

\newcommand{\nD}{\mathbf{D}}      

\newcommand{\avec}{{\vec a}}      
\newcommand{\xvec}{{\vec x}}      
\newcommand{\yvec}{{\vec y}}      

\newcommand{\DD}{\mathfrak{D}}    
\newcommand{\g}{\mathfrak{g}}     
\newcommand{\gsu}{\mathfrak{su}}  

\newcommand{\id}{\mathsf{id}}     

\newcommand{\az}{\triangleright}  
\newcommand{\Dirac}{D\mkern-11.5mu/\,} 

\newcommand{\isom}{\simeq}        
\newcommand{\hookto}{\hookrightarrow} 
\newcommand{\ovl}{\overline}      
\newcommand{\ox}{\otimes}         
\newcommand{\wh}{\widehat}        
\newcommand{\wt}{\widetilde}      
\newcommand{\x}{\times}           
\renewcommand{\.}{\cdot}          
\renewcommand{\:}{\colon}         

\newcommand{\half}{{\mathchoice{\thalf}{\thalf}{\shalf}{\shalf}}}
\newcommand{\quarter}{\tfrac{1}{4}} 
\newcommand{\pihalf}{\frac{\pi}{2}} 
\newcommand{\shalf}{{\scriptstyle\frac{1}{2}}} 
\newcommand{\thalf}{\tfrac{1}{2}} 

\newcommand{\cket}[1]{\lvert#1)}      
\newcommand{\ket}[1]{\lvert#1\rangle} 
\newcommand{\kket}[1]{\lvert#1\rangle\!\rangle} 
\newcommand{\set}[1]{\{\,#1\,\}}  
\newcommand{\word}[1]{\quad\mbox{#1}\quad} 

\newcommand{\cbraket}[2]{(#1\mathbin|#2)} 
\newcommand{\pd}[2]{\frac{\partial#1}{\partial#2}} 
\newcommand{\scal}[2]{\langle#1\mathbin|#2\rangle} 

\newcommand{\braCket}[3]{\langle#1\mathbin|#2\mathbin|#3\rangle} 
\newcommand{\cbraCket}[3]{(#1\mathbin|#2\mathbin|#3)} 

\newcommand{\twobytwo}[4]{\begin{pmatrix}#1 & #2 \\
                            #3 & #4\end{pmatrix}} 

\newcommand{\CG}[6]{C_{#1#4,#2#5}^{#3#6}} 

\theoremstyle{plain}
\newtheorem{prop}{Proposition}[section]  
\newtheorem{lemma}[prop]{Lemma}          
\newtheorem{corl}[prop]{Corollary}       

\theoremstyle{definition}
\newtheorem{defn}[prop]{Definition}      

\theoremstyle{remark}
\newtheorem{remk}[prop]{Remark}          

\numberwithin{equation}{section}

\newcommand{\hideqed}{\renewcommand{\qed}{}} 

\makeatletter
\renewcommand{\section}{\@startsection{section}{1}{\z@}%
                        {-3.5ex \@plus -1ex \@minus -.2ex}%
                        {2.3ex \@plus.2ex}%
                        {\normalfont\large\bfseries}}
\renewcommand{\subsection}{\@startsection{subsection}{2}{\z@}%
                        {-3.25ex \@plus -1ex \@minus -.2ex}%
                        {1.5ex \@plus .2ex}%
                        {\normalfont\normalsize\bfseries}}
\makeatother

\begin{document}

\setlength{\droptitle}{-3pc}  
\pretitle{\begin{flushright}\small
ICCUB--12--320                
\end{flushright}
\begin{center}\LARGE}
\posttitle{\par\end{center}}

\maketitle

\vspace*{-1pc}

\begin{abstract}
The fuzzy sphere, as a quantum metric space, carries a sequence of 
metrics which we describe in detail. We show that the Bloch coherent
states, with these spectral distances, form a sequence of metric 
spaces that converge to the round sphere in the high-spin limit.
\end{abstract}

\footnotetext{%
\textit{MSC class} [2010]: Primary 46L87; secondary 81R05, 58B34, 81R30.}

\footnotetext{%
\textit{Keywords}: fuzzy sphere, quantum metric space, Dirac
operators, coherent states.}


\section{Introduction} 

It is common practice in several fields to ``approximate'' a manifold
with a finite or countable subset of points. A typical example in is
the study of quantum field theories on a lattice. One drawback is the
absence of some of the symmetries of the continuous theory it purports
to approximate (e.g., Poincar\'e symmetries in flat Minkowski space).

Take the simple example of a unit two-sphere $\bS^2$. On replacing
$\bS^2$ with a subset of $N$ points, rotational symmetry is lost. In
algebraic language: the algebra $\bC^N$ of functions on $N$ points is
not an $\sU(\gsu(2))$-module $*$-algebra. There are no nontrivial
$SU(2)$-orbits with finitely many points; to preserve the symmetries
and keep the algebra finite dimensional, one may replace the function
algebra $\bC^N$ with a noncommutative one, provided that the
noncommutativity be suppressed as $N \to \infty$. This is the idea
behind the fuzzy sphere (and more general fuzzy spaces), put forward
in~\cite{Mad92}, as well as in \cite{Strat57,Hop82,VG89}.

Let $x_1,x_2,x_3$ be Cartesian coordinates on $\bS^2$, and
$\sA(\bS^2)$ be the $*$-algebra of poly\-no\-mials in these. As an
abstract $*$-algebra, this is the complex unital commutative
$*$-algebra with three self-adjoint generators $x_1$, $x_2$, $x_3$
subject only to the relation $x_1^2 + x_2^2 + x_3^2 = 1$. As an
$\sU(\gsu(2))$-module $*$-algebra, $\sA(\bS^2)$ decomposes into a
direct sum of irreducible representations
$\sA(\bS^2) \isom \bigoplus_{\ell=0}^\infty V_\ell$. Here $V_\ell$ is
the vector space underlying the irreducible representation of
$\sU(\gsu(2))$ with highest weight $\ell \in \bN$, and is spanned by
Laplace spherical harmonics $Y_{\ell,m}$.

In the spirit of \cite{AL10,AKL11}, we introduce a cut-off in the
energy spectrum, i.e., we neglect all but the first $(N + 1)$
representations in the decomposition of $\sA(\bS^2)$. One cannot
simply take the linear span of $Y_{\ell,m}$ for $\ell = 0,1,\dots,N$,
as this is not a subalgebra of $\sA(\bS^2)$. To proceed, we write
$N = 2j$ and denote by $\pi_j \: \sU(\gsu(2)) \to M_{2j+1}(\bC)$ the
spin~$j$ representation of $\sU(\gsu(2))$; the action (using
Sweedler notation for the coproduct):
$$
h \az a := \pi_j(h_{(1)})\, a \,\pi_j(S(h_{(2)})),  \qquad
h \in \sU(\gsu(2)),\ a \in M_{2j+1}(\bC)
$$
makes the matrix algebra $\sA_N := M_{N+1}(\bC)$ an
$\sU(\gsu(2))$-module $*$-algebra. There is a decomposition into
irreducible representations:
$$
\sA_N \isom V_j \ox V_j^* \isom \bigoplus_{\ell=0}^{2j} V_\ell
$$
and a surjective homomorphism $\sA(\bS^2) \to \sA_N$ of
$\sU(\gsu(2))$-modules (but not of module algebras), given on
generators by
$$
x_k \mapsto \hat x_k := \frac{1}{\sqrt{j(j + 1)}}\, \pi_j(J_k),
$$
where the $J_k$ are the standard real generators of $\sU(\gsu(2))$. 
The map $x_k \mapsto \hat x_k$ does not extend to an
algebra morphism, but can be extended in a unique way, using
coherent-state quantization, to an isometry between
$*$-representations of $\sU(\gsu(2))$ sending the spherical harmonic
$Y_{\ell,m}$, for $\ell \leq 2j$, into a matrix $\wh Y_{\ell,m}^{(j)}$
sometimes called a ``fuzzy spherical harmonic'' (details
at the end of Sect.~\ref{ssc:Dirac-fuzzy}).
Since an infinite-dimensional vector space is mapped onto a
finite-dimensional one, 
information is lost and the space becomes ``fuzzy''.

The matrices $\hat x_k$ are normalized in such a way that the
spherical relation still holds:
$\hat x_1^2 + \hat x_2^2 + \hat x_3^2 = 1$, but their commutators are
clearly not zero~\cite{VG89}:
$$
[\hat x_k,\hat x_l]
= \frac{1}{\sqrt{j(j + 1)}}\, \ii\,\eps_{klm} \hat x_m.
$$
Since the coefficient in the commutator vanishes for
$N = 2j \to \infty$, the na\"ive idea is that the fuzzy sphere
``converges'', as $N \to \infty$, to a unit sphere. It is clear that
the notion of convergence must involve the Riemannian metric
of~$\bS^2$.

The correct mathematical framework for the convergence of matrix
algebras to algebras of functions on Riemannian manifolds (or more
generally, on metric spaces) was developed by Rieffel in a series of
seminal papers, where he introduced the notion of (compact) quantum
metric spaces and quantum Gromov--Hausdorff convergence
\cite{Rie99,Rie04,Rie04a}. The convergence of the fuzzy sphere
to~$\bS^2$ was established in~\cite{Rie04b}. However, there the
metrics are dealt with globally and the proof does not indicate how to
choose a sequence of elements approximating a given point of~$\bS^2$.
In this paper we approximate the points of~$\bS^2$ by the
corresponding (Bloch) coherent states of~$\sA_N$.

A distance $d_N$ on the state space of $\sA_N$ can be defined via a
generalized Dirac operator. Since, for any~$N$, the set of coherent
states is labelled by~$\bS^2$, this gives a distance on $\bS^2$
depending on the deformation parameter~$N$. Denoting by $d_\geo$ the
geodesic distance of the round sphere, we prove that
$$
\lim_{N\to\infty} d_N(p,q) = d_\geo(p,q), \word{for all} p,q\in \bS^2.
$$

Another noncommutative space where the distance between coherent
states has already been studied is the Moyal plane
\cite{CDMW11,MT11,Wal11}. In contrast with that example, whose
distance is independent of the deformation parameter, here the
distance depends on~$N$.

\vspace{6pt}  

Sect.~\ref{sec:NCspaces} briefly recalls the basics of noncommutative
spaces. In Sect.~\ref{sec:Dirac}, we introduce our spectral triples
for the sphere and compare them with other proposals in the
literature. In Sect.~\ref{sec:dist}, we recall the Bloch coherent
states~\cite{ACGT72} and compute some particular distances between
them. Then we prove that the spectral distance is $SU(2)$-invariant,
nondecreasing with~$N$, and converges to the geodesic distance on
$\bS^2$ when $N \to \infty$.


\section{Preliminaries on noncommutative manifolds} 
\label{sec:NCspaces}

Material in this section is mainly taken from \cite{Con94,GVF01}. In
the spirit of Connes' noncommutative geometry, manifolds are replaced
by spectral triples.
A \emph{unital spectral triple} $(\sA,\sH,D)$ has the following data:
(i) a separable complex Hilbert space $\sH$;
(ii) a complex associative involutive unital algebra $\sA$ with a
faithful unital $*$-representation $\sA \to \sB(\sH)$, the
representation symbol usually being omitted;
(iii) a self-adjoint operator $D$ on~$\sH$ with compact resolvent
such that $[D,a]$ is a bounded operator for all $a \in \sA$.

A spectral triple is \emph{even} if there is a \emph{grading}
$\ga$ on $\sH$, i.e., a bounded operator satisfying $\ga = \ga^*$ and
$\ga^2 = 1$, commuting with any $a \in\sA$ and anticommuting with~$D$.

A spectral triple is \emph{real} if there is an antilinear isometry
$J \: \sH \to \sH$ (the ``real structure''), such that $J^2 = \pm 1$,
$JD = \pm DJ$ and $J\ga = \pm\ga J$ in the even case, with the signs
related to the KO-dimension of the triple~\cite{Con95}; and
\begin{equation}
\label{eq:real-comm} 
[a,JbJ^{-1}] = 0,  \qquad  [[D,a],JbJ^{-1}] = 0,
\quad\word{for all} a,b \in \sA.
\end{equation}
This shows that $b \mapsto Jb^*J^{-1}$ is an injective homomorphism of
$\sA$ into its commutant.

For the notion of \emph{equivariant} spectral triple, we refer
to~\cite{Sit03}. A group, or more generally a Hopf algebra, acts
on~$\sA$ and on~$\sH$, intertwining the operator $D$ with itself.

\begin{remk} 
\label{rk:even}
Note that if $(\sA,\sH,D,\ga)$ is an even spectral triple and $v$ an
eigenvector of $D$ with eigenvalue $\la$, then $\ga v$ is an
eigenvector of $D$ with eigenvalue $-\la$. Thus, the eigenvalues $\la$
and $-\la$ have the same multiplicity.
\end{remk}

We use the following notations and conventions. $\sB(\sH)$ is the
algebra of all bounded linear operators on~$\sH$. The set of all
states of (the norm completion of) $\sA$ is denoted by $\sS(\sA)$.
We denote by $\|\.\|$ the operator norm of $\sB(\sH)$; by
$\|v\|_\sH^2 = \scal{v}{v}$ the norm-squared of a vector $v \in \sH$,
writing $\scal{\.}{\.}$ for scalar products. By $\Cl(\g)$ we mean the
Clifford algebra over a semisimple Lie algebra with its Killing form.

Recall that $\sS(\sA)$ is a convex set, compact in the weak$^*$ 
topology, whose extremal points are the \emph{pure states} of~$\sA$.
$\sS(\sA)$ is an extended metric space (allowing distances to be 
$+\infty$), with distance function given by
\begin{equation}
\label{eq:metric} 
d_{\sA,D}(\om,\om') := \sup_{a=a^*\in\sA}
\set{|\om(a) - \om'(a)| : \|[D,a]\| \leq 1}
\end{equation}
for all $\om,\om'\in \sS(\sA)$. This is usually called \emph{Connes'
metric} or \emph{spectral distance}~\cite{Con89}. The supremum is
usually taken over all $a \in \sA$ obeying the side condition; but it
was noted in~\cite{IKM01} that the supremum is always attained on
self-adjoint elements. More generally, when defining a metric, one can
replace $\|[D,a]\|$ by $L(a)$ where $L$ is a Leibniz seminorm
on~$\sA$. The structure $(\sA,d_{\sA,L})$ so obtained is a ``compact
quantum metric space'' \cite{Rie04,Rie10}.


\section{Dirac operators for the fuzzy sphere} 
\label{sec:Dirac}

The classical Dirac operator $\Dirac$ on a compact semisimple Lie
group $G$ with Lie algebra $\g$ can be seen as a purely algebraic
object $\DD$ living in the noncommutative Weil algebra
$U(\g) \ox \Cl(\g)$, see~\cite{Kos99,HP06}. It is equivariant in the
sense that there exists a Lie algebra homomorphism
$\g \to U(\g) \ox \Cl(\g)$ with whose range $\DD$ commutes. The spinor
bundle of $G$ is parallelizable: $L^2(G,S) \isom L^2(G) \ox \Sg$,
where $\Sg$ is an irreducible $ \Cl(\g)$-module. The algebra
$U(\g) \ox \Cl(\g)$ acts on the Hilbert space $L^2(G) \ox \Sg$ making
$\DD$ into the ``concrete'' Dirac operator $\Dirac$ of $G$, an
unbounded first-order elliptic operator. Using the injection
$\g \hookto \Cl(\g)$ we can also think of $\DD$ as an element of
$U(\g) \ox U(\g)$, equivariant in the sense that it commutes with the
range of the coproduct $\Delta$ in $U(\g) \ox U(\g)$. On a compact
Riemannian symmetric space $G/U$, this construction also applies
(indeed, it works on~$G$ as a symmetric space of $G \x G$), although
the spinor bundle is not always parallelizable. This is the point of
view that we shall adopt for the fuzzy sphere.

\subsection{An abstract Dirac operator} 
\label{ssc:Dabs-S2}

We begin with the two-sphere $\bS^2$. 
The abstract \textit{Dirac element}
$\DD \in \sU(\gsu(2)) \ox \sU(\gsu(2))$ is defined as
\begin{equation}
\label{eq:D-elt} 
\DD := 1 \ox 1 + 2 \tsum_k J_k \ox J_k \,.
\end{equation}
Since $\sum_k [J_k \ox J_k, J_l \ox 1 + 1 \ox J_l] = 0$, this element
commutes with the range of the coproduct
$\Dl \: \sU(\gsu(2)) \to \sU(\gsu(2)) \ox \sU(\gsu(2))$. That is an
equivariance property of~$\DD$.

The corresponding element of $\sU(\gsu(2)) \ox \Cl_{20}$ is
\begin{equation}
\label{eq:abs-D} 
\DD_S := (\id \ox \pi_\half)(\DD) = 1 \ox 1 + \tsum_k J_k \ox \sg_k
= \twobytwo{1 + H}{F}{E}{1 - H}.
\end{equation}
where $H = J_3$, $E = J_1 + \ii J_2$, $F = E^*$. The square of~$\DD$
is $\DD_S^2 = C_{SU(2)} + \quarter(1 \ox 1)$, where
$C_{SU(2)} := \sum_k (J_k \ox 1 + 1 \ox \half\sg_k)^2$ is the Casimir
operator and $1/4 = R/8$ is the scalar curvature term ($R = 2$ being
the scalar curvature of~$\bS^2$). This is the symmetric space version,
$D^2 = C_G + R/8$, of the Schr\"odinger--Lichnerowicz formula for
equivariant Dirac operators \cite[p.~87]{Fri00}.

\begin{lemma} 
\label{lm:eigenD}
For any $\ell \neq 0$ in $\half\bN$, the operator
$(\pi_\ell \ox \pi_\half)(\DD^2)$ has eigenvalues $\ell^2$ with
multiplicity $2\ell$ and $(\ell + 1)^2$ with multiplicity $2\ell + 2$.
For $\ell = 0$, $(\pi_0 \ox \pi_\half)(\DD^2)$ has eigenvalue~$1$ with
multiplicity~$2$.
\end{lemma}

\begin{proof}
With $J^2 = \sum_k J_k^2$, it follows from
$\Dl(J^2) = \sum_k \Dl(J_k)^2 = \sum_k (J_k \ox 1 + 1 \ox J_k)^2$ that
$C_{SU(2)} = (\id \ox \pi_\half) \Dl(J^2)$. Since $\Dl(1) = 1 \ox 1$,
this yields $\DD_S^2 = (\id \ox \pi_\half)\Dl(J^2 + \quarter)$.
Therefore,
$$
(\pi_\ell \ox \pi_\half)(\DD^2) = (\pi_\ell \ox \id)(\DD_S^2)
= (\pi_\ell \ox \pi_\half)\Dl(J^2 + \quarter).
$$
Now $(\pi_\ell \ox \pi_\half)\Dl$ is the Hopf tensor product of the
representations $\pi_\ell$ and $\pi_\half$. From
\begin{equation}
\label{eq:CGhalf} 
V_\ell \ox V_\half \isom  V_{\ell+\half} \oplus V_{\ell-\half}
\end{equation}
it follows that $(\pi_\ell \ox \pi_\half)(\DD^2)$ is unitarily
equivalent to $\pi_{\ell+\half}(J^2 + \quarter)
\oplus \pi_{\ell-\half}(J^2 + \quarter)$, and hence has eigenvalues
$$
(\ell \pm \half)(\ell \pm \half + 1) + \quarter = \begin{cases}
(\ell + 1)^2 & \text{on } V_{\ell+\half}\,, \\
\ell^2 & \text{on } V_{\ell-\half}\,.
\end{cases}
$$
If $\ell = 0$, the summand $V_{\ell-\half}$ in \eqref{eq:CGhalf} is
missing, so the only eigenvalue is~$1$ on~$V_\half$.
\end{proof}

\subsection{The Dirac operator of $\bS^2$} 
\label{ssc:Dirac-S2}

The natural representation of $\sU(\gsu(2))$ on $\bS^2$ as vector
fields yields the Dirac operator $\Dirac$ of the unit sphere (with
round metric). The spinor bundle $S \to \bS^2$ is trivial of 
rank~$2$, so the spinor space is
$L^2(\bS^2,S) \isom L^2(\bS^2) \ox \bC^2$.

Modulo the identification
$L^2(\bS^2) \isom \bigoplus_{\ell\in\bN} V_\ell$, the operator
$\Dirac$ is given by
$$
\Dirac = \bigoplus_{\ell\in\bN} (\pi_\ell \ox \pi_\half)(\DD),
$$
with $\DD$ as in~\eqref{eq:D-elt}. It follows from
Lemma~\ref{lm:eigenD} that $\Dirac^2$ has eigenvalues
$\la_\ell = \ell^2$ with multiplicity $m_\ell = 4\ell$, for every
integer $\ell \geq 1$. The spectral triple of $\bS^2$ is even, using
the grading that exchanges the two half-spinor line
bundles~\cite{GVF01}. From Remark~\ref{rk:even} it follows that
$\Dirac$ has eigenvalues $\pm\ell$ with multiplicities
$\half m_\ell = 2\ell$.

\subsection{Dirac operators on the fuzzy sphere} 
\label{ssc:Dirac-fuzzy}

We require an equivariant Dirac operator whose spectrum is that
of~$\Dirac$, truncated at $\ell = N + 1$. Let $N = 2j \geq 1$ be a
fixed integer. The \textit{fuzzy sphere} (labelled by~$N$) is the
``noncommutative $SU(2)$ coset space'' described by the algebra
$\sA_N := M_{N+1}(\bC)$ with the $SU(2)$ left action
$(g,a) \mapsto a^g := \pi_j(g)\,a\,\pi_j(g)^*$, for $g \in SU(2)$,
$a \in \sA_N$.

\begin{defn} 
\label{df:DN-irred}
The \emph{irreducible} spectral triple on $\sA_N$, that we denote by
$(\sA_N,\sH_N,D_N)$, is given by $\sH_N := V_j \ox \bC^2$, with the
natural representation of $\sA_N$ via row-by-column multiplication on
the factor $V_j \isom \bC^{N+1}$, and
$D_N := (\pi_j \ox \pi_\half)(\DD)$, where $\DD$ is the abstract Dirac
element in~\eqref{eq:D-elt}.
\end{defn}

\begin{prop} 
\label{pr:DN-irred}
The irreducible spectral triple on $\sA_N = \sA_{2j}$ has these
properties:
\begin{enumerate}
\item[\textup{(i)}]
It is equivariant with respect to the $SU(2)$ representation
$\pi_j \ox \pi_\half\,$.
\item[\textup{(ii)}]
$D_N$ has eigenvalues $j + 1$ and~$(-j)$, with respective
multiplicities $2j + 2$ and~$2j$.
\item[\textup{(iii)}]
No grading or real structure is compatible with this spectral triple.
\end{enumerate}
\end{prop}

\begin{proof}
Equivariance comes from the commuting of $\DD_S$ with the range of the
coproduct, so that $D_N$ commutes with the representation
$\pi_j \ox \pi_\half$ of $\sU(\gsu(2))$ ---or the corresponding
representation of $SU(2)$--- and from the intertwining relation:
$$
(\pi_j \ox \pi_\half)(g)\,(a \ox 1)\,(\pi_j \ox \pi_\half)(g)^*
= \pi_j(g)\,a\,\pi_j(g)^* \ox \pi_\half(g)\pi_\half(g)^* = a^g \ox 1,
\quad a \in \sA_N \,.
$$

{}From Lemma~\ref{lm:eigenD} follows that $D_N^2$ has eigenvalues
$j^2$ and $(j + 1)^2$. However, the spectrum of $D_N$ is \emph{not}
symmetric about~$0$. Indeed, an explicit computation shows that
$\sH_N$ has the following orthonormal basis of eigenvectors for~$D_N$:
\begin{align}
\kket{j,m}_+
&:= \sqrt{\tfrac{j+m+1}{2j+1}}\, \ket{j,m} \ox \tbinom{1}{0}
+ \sqrt{\tfrac{j-m}{2j+1}}\, \ket{j,m+1} \ox \tbinom{0}{1},
& m &= -j - 1, \dots, j \,;
\nonumber \\[2\jot]
\kket{j,m}_-
&:= - \sqrt{\tfrac{j-m}{2j+1}}\, \ket{j,m} \ox \tbinom{1}{0}
+ \sqrt{\tfrac{j+m+1}{2j+1}}\, \ket{j,m+1} \ox \tbinom{0}{1}, 
& m &= -j, \dots, j - 1 \,.
\label{eq:spinors} 
\end{align}
One easily checks that
$$
D_N \ket{j,m}_+ = (j + 1) \ket{j,m}_+ \,,  \qquad
D_N \ket{j,m}_- = -j \ket{j,m}_- \,.
$$
Therefore, $D_N$ has eigenvalue $j + 1$ with multiplicity $2j + 2$,
and eigenvalue $-j$ with multiplicity $2j$, as claimed. This asymmetry
of the spectrum of $D_N$ and Remark~\ref{rk:even} rule out any grading
for this spectral triple.

If there were a real structure, the commutant $\sA'_N$ of $\sA_N$
would contain $J \sA_N J^{-1}$, whose dimension is $(N + 1)^2 \geq 4$.
But $\dim \sA'_N = 2$; hence, no real structure can exist.
\end{proof}

\begin{defn} 
\label{df:DN-full}
The \emph{full} spectral triple on $\sA_N$, that we denote by
$(\sA_N, \wt\sH_N, \wt\sD_N, \wt\sJ_N)$, is given by
$\wt\sH_N \isom \sA_N \ox \bC^2$, where the first factor carries the
left regular representation of $\sA_N$, i.e., the GNS representation
associated to the matrix trace; and the Dirac operator and real
structure are defined by:
\begin{align*}
\wt\sD_N(a \ox v) &:= a \ox v + \tsum_k [\pi_j(J_k), a] \ox \sg_k v,
\\
\wt\sJ_N(a \ox v) &:= a^* \ox \sg_2 \bar v,
\end{align*}
for any $a \in \sA_N$ and $v \in \bC^2$ (a column vector). For
$v = (v_1,v_2)^t \in \bC^2$, $\bar v := (v_1^*,v_2^*)^t$ is again a
column vector.
\end{defn}

The nuance between $D_N$ and $\wt\sD_N$ is that $\pi_j$ is replaced by
its adjoint action on the space
$\sA_N = \End(V_j) \isom V_j \ox V_j^*$.

\begin{prop} 
\label{pr:DN-full}
The full spectral triple on $\sA_N$ has the following properties:
\begin{enumerate}
\item[\textup{(i)}]
It is a real spectral triple.
\item[\textup{(ii)}]
It is equivariant with respect to the $SU(2)$ representation given by
the product of the action $a \mapsto a^g$ on~$\sA_N$ and the
spin-$\half$ representation.
\item[\textup{(iii)}]
$\wt\sD_N$ has integer eigenvalues $\pm\ell$ with multiplicity
$2\ell$, for every $\ell = 1,\dots,N$; and eigenvalue $N + 1$ with
multiplicity $2N + 2$.
\item[\textup{(iv)}] 
This spectral triple carries no grading.
\end{enumerate}
\end{prop}

\begin{proof}
Clearly $\wt\sJ_N$ is antilinear, indeed antiunitary, since
\begin{align*}
\scal{\wt\sJ_N(a \ox v)}{\wt\sJ_N(b \ox w)}
&= \Tr(a^*b) \scal{\sg_2\bar v}{\sg_2\bar w}
= \Tr(a^*b) \scal{\bar v}{\bar w}
\\
&= \ovl{\Tr(b^*a)}\,\,\ovl{\scal{w}{v}}
= \ovl{\scal{b \ox w}{a \ox v}} \,.
\end{align*}
We need to check the conditions~\eqref{eq:real-comm}. The equality
$\bar\sg_2 = - \sg_2$ shows that $(\wt\sJ_N)^2 = -1$. Using
$\wt\sJ_N^{-1} = - \wt\sJ_N$, we find that
$$
\wt\sJ_N\,b\,\wt\sJ_N^{-1}(a \ox v) = -\wt\sJ_N(ba^* \ox \sg_2\bar v)
= ab^* \ox v,  \word{for all}  a,b \in \sA_N, \ v \in \bC^2.
$$
Since left and right multiplication on $\sA_N$ commute,
$\wt\sJ_N\,b\,\wt\sJ_N^{-1}$ lies in the commutant of
$\sA_N \ox M_2(\bC)$, and both conditions in~\eqref{eq:real-comm} are
satisfied.

Since $\sg_2 \bar\sg_k = - \sg_k\sg_2$ for $k = 1,2,3$, and 
$[\pi_j(J_k), a]^* = -[\pi_j(J_k), a^*]$, we obtain
$$
\wt\sJ_N \wt\sD_N (a \ox v)
= a^* \ox \sg_2 v - \tsum_k [\pi_j(J_k), a^*] \ox \sg_2 \bar\sg_k v 
= \wt\sD_N \wt\sJ_N (a \ox v),
$$
for any $a \in \sA_N$ and $v \in \bC^2$. Hence
$\wt\sJ_N \wt\sD_N = \wt\sD_N \wt\sJ_N$.

Equivariance follows again from the commuting of $\DD_S$ with the
range of the coproduct, since the representation
$J_k \mapsto [\pi_j(J_k), \cdot]$ is the derivative of the adjoint
action $a \mapsto a^g = \pi_j(g)\,a\,\pi_j(g)^*$ of~$SU(2)$.

Writing $\ad\pi_j(h) \: a \mapsto \pi_j(h_{(1)})\,a\,\pi_j(S(h_{(2)}))$
for $h \in \sU(\gsu(2))$ and $a \in \sA_N$, we see that
$$
\wt\sD_N = (\ad\pi_j \ox \pi_\half)(\DD),
$$
In view of the unitary $\sU(\gsu(2))$-module isomorphism
$$
\sA_N \isom  V_j \ox V_j^* \isom \bigoplus_{\ell=0}^{2j} V_\ell \,,
$$
$\wt\sD_N$ is unitarily equivalent to the operator
$\bigoplus_{\ell=0}^{2j}(\pi_\ell \ox \pi_\half)(\DD)$.
Replacing $N = 2j$ by~$2\ell$ in Prop.~\ref{pr:DN-irred}(ii), we see
that $(\pi_\ell \ox \pi_\half)(\DD)$ has eigenvalues $\ell + 1$
and~$(-\ell)$, with respective multiplicities $2\ell + 2$ and~$2\ell$
(but if $\ell = 0$ the eigenvalue $-\ell$ is missing). Hence
$\wt\sD_N$ has the eigenvalues $\pm\ell$, each with multiplicity
$2\ell$ for $\ell = 1,\dots,N$; and $N + 1$ with multiplicity
$2N + 2$.

Lastly, since the spectrum of $\wt\sD_N$ is not symmetric about~$0$,
there can exist no grading for this spectral triple.
\end{proof}

\begin{prop} 
\label{pr:same-metric}
The irreducible and full spectral triples induce the same metric
on the state space $\sS(\sA_N)$ of the fuzzy sphere.
\end{prop}

\begin{proof}
This follows from the calculation:
\begin{align*}
\qquad
[\wt\sD_N, a](b \ox v) &= \tsum_k \bigl( [\pi_j(J_k), ab]
- a\,[\pi_j(J_k), b] \bigr) \ox \sg_k v
\\[\jot]
&= \tsum_k [\pi_j(J_k), a]\,b \ox \sg_k v = [D_N,a](b \ox v).
\end{align*}
Hence $[\wt\sD_N, a]$ is the operator of left multiplication by the
matrix $[D_N,a] \in \sA_N \ox M_2(\bC)$, so its operator norm coincides
with the norm of the matrix. Therefore, since
$\bigl\| [\wt\sD_N,a] \bigr\| = \bigl\| [D_N,a] \bigr\|$ for each $a
\in \sA_N$, it follows that the two spectral triples induce the same
metric \eqref{eq:metric} on the state space of~$\sA_N$.
\end{proof}

It is useful to give a more explicit presentation of the full 
spectral triple by exhibiting its eigenspinors. Recall that the 
polynomial algebra $\sA(\bS^2)$ is linearly spanned by the spherical 
harmonics $Y_{\ell,m}$, each of which is a homogeneous polynomial in 
Cartesian coordinates of degree~$\ell$, with the multiplication rule
$$
Y_{\ell',m'} Y_{\ell'',m''}
= \sum_{\ell=|\ell'-\ell''|}^{\ell=\ell'+\ell''} \sum_{m=-\ell}^\ell
\sqrt{\frac{(2\ell' + 1)(2\ell'' + 1)}{4\pi(2\ell + 1)}}\,
\CG{\ell'}{\ell''}{\ell}{0}{0}{0} 
\CG{\ell'}{\ell''}{\ell}{m'}{m''}{m} Y_{\ell,m} \,,
$$
involving $SU(2)$ Clebsch--Gordan coefficients. From there it is 
clear that the subspace spanned by the $Y_{\ell,m}$ for 
$\ell = 0,1,\dots,N$ does not close under multiplication. To replace 
them, while keeping $SU(2)$ symmetry, one can make use of the 
irreducible tensor operators at level $N = 2j$ \cite{Aga81,BL81}.
These are elements $\wh T^{(j)}_{\ell,m} \in M_{N+1}(\bC)$ whose
matrix elements are
$$
\braCket{jm''}{\wh T^{(j)}_{\ell,m}}{jm'}
:= \sqrt{\frac{2\ell+1}{2j+1}}\, \CG{j}{\ell}{j}{m'}{m}{m''} \,.
$$
They transform like the $Y_{\ell,m}$ under $SU(2)$, but still require
an appropriate normalization. For any $-1 \leq s \leq 1$, one can
define a matrix $\wh Y^{(j,s)}_{\ell,m} \in M_{N+1}(\bC)$ as follows
\cite{AC91,BM98,KE02,CMS01}:
\begin{equation}
\label{eq:fuzzy-harms} 
\wh Y^{(j,s)}_{\ell,m} := \sqrt{\frac{4\pi}{2j+1}}\, 
\bigl( \CG{j}{\ell}{j}{j}{0}{j} \bigr)^s \,\wh T^{(j)}_{\ell,m}.
\end{equation}
We omit the precise multiplication rules for these operators,
see~\cite{KE02}; but in any case it is clear, by working backwards,
that the ordinary spherical harmonic $Y_{\ell, m}$ can be
regarded as a ``symbol'' of the operator $\wh Y^{(j,s)}_{\ell,m}$ for
fixed $j$ and~$s$. The cases $s = 1$, $s = 0$ and $s = -1$ correspond
respectively to the Husimi $Q$-function, the Moyal--Wigner
$W$-function and the Glauber $P$-function~\cite{BM98}. Here we put
$s = 1$ in~\eqref{eq:fuzzy-harms}, omit the superscripts, and call
these operators the \textit{fuzzy harmonics}
$\wh Y_{\ell,m} \in\sA_N$. The commutation rules for the irreducible
tensor operators and the fuzzy harmonics come directly from their
symmetries~\cite{Aga81,BL81,VG89}:
\begin{align*}
[\pi_j(J_3), \wh Y_{\ell,m}] 
&= [\wh Y_{1,0}, \wh Y_{\ell,m}] = m\,\wh Y_{\ell,m} \,,
\\
[\pi_j(J_1 \pm \ii J_2), \wh Y_{\ell,m}]
&=  [\wh Y_{1,\pm 1}, \wh Y_{\ell,m}]
= \sqrt{(\ell \mp m)(\ell \pm m + 1)}\,\wh Y_{\ell,m\pm 1} \,.
\end{align*}
Adopting a $2 \x 2$ block matrix notation, as in~\eqref{eq:abs-D}, we
can write
\begin{equation}
\wt\sD_N = \twobytwo{1 + \sL_3}{\sL_-}{\sL_+}{1 - \sL_3}, \word{where}
\sL_3 = \ad\pi_j(J_3), \ \sL_\pm = \ad\pi_j(J_1 \pm \ii J_2).
\label{eq:full-Dirac} 
\end{equation}
Then the normalized eigenspinors for the operators $\wt\sD_N$ are
$$
\kket{\ell,m}_+ := \frac{1}{\sqrt{2\ell+1}} \begin{pmatrix}
\sqrt{\ell+m+1}\, \wh Y_{\ell,m} \\[\jot]
\sqrt{\ell-m}\, \wh Y_{\ell,m+1} \end{pmatrix},  \quad
\kket{\ell,m}_- := \frac{1}{\sqrt{2\ell+1}} \begin{pmatrix}
- \sqrt{\ell-m}\, \wh Y_{\ell,m} \\[\jot]
\sqrt{\ell+m+1}\, \wh Y_{\ell,m+1} \end{pmatrix}
$$
for $\ell = 0,1,\dots,N$; whereby
\begin{align*}
\wt\sD_N \kket{\ell,m}_+ &= (\ell + 1)\,\kket{\ell,m}_+
\word{for}  m = -\ell-1,\dots,\ell\,;
\\
\wt\sD_N \kket{\ell,m}_- &= \quad (-\ell)\,\kket{\ell,m}_-
\word{for}  m = -\ell,\dots,\ell - 1.
\end{align*}

The full spectral triple on~$\sA_N$ is thus a truncation of the
standard spectral triple over~$\bS^2$, in the following sense. The
Hilbert space of spinors $L^2(\bS^2) \ox \bC^2$, generated by pairs of
spherical harmonics $Y_{\ell,m}$, is truncated at $l \leq N$. On
replacing these by pairs of fuzzy harmonics $\wh Y_{\ell,m}$, the
resulting spectrum of $\wt\sD_N$ is a truncation of the spectrum of
$\Dirac$ to the range $\{-N, \dots, N + 1\}$, unavoidably breaking the
parity symmetry.

\subsection{Comparison with the literature} 
\label{ssc:comp-lit}

Two spectral triples on the fuzzy sphere algebra $\sA_n$ have been
introduced, one constructed with the irreducible $\gsu(2)$-module
$V_j$ and the other with the left regular or GNS representation.
Neither one is even (there exists no grading); although this could be
remedied by allowing $\sA_N$ to act trivially on a supplementary
vector space. The first carries no real structure but the second one
does, because the reducible action of the algebra on the Hilbert space
allows for a large enough commutant. The crucial point here, however,
is Prop.~\ref{pr:same-metric}, showing that both spectral triples give
the same metric. Other Dirac operators for the fuzzy sphere have been
proposed in \cite{GKP97,CWW97,CWW98,BP07,BP09} and are recalled below.

In \cite{GKP97}, $\sA_N$ is obtained as the even part of a truncated
supersphere, and the Dirac operator is defined as the odd part of a
truncated superfield. Reformulating the result of Sect.~4.3
of~\cite{GKP97} in our language, the Hilbert space is taken to be
$$
\sH'_N := \bigoplus_{\ell=\half,\dots,N-\half} V_\ell \oplus V_\ell\,.
$$
Note that to get our $\sA_N \ox V_\half$ one must add an extra
$V_{N+\half}$ subspace. The algebra $\sA_N$ is generated by the three
matrices $\hat x_k$, proportional to $\pi_j(J_k)$, which can be
represented on $\sH'_N$ using a suitable direct sum of irreducible
representations of $\gsu(2)$. The Dirac operator can be defined by
representing the abstract Dirac element \eqref{eq:D-elt} on $\sH'$
using the same representation of $\gsu(2)$; it is proportional to the
identity on each subspace $V_\ell$ and its spectrum is given by the
eigenvalues $\pm\ell$, for $\ell = 1,\dots,N$ (restricted to
$V_\ell \oplus V_\ell$ their Dirac operator is the operator
$\ell \oplus -\ell$). Compared to our full spectral triple, the
eigenvalue $N + 1$ is missing. Since the two copies of $V_\ell$ carry
the same representation of $\sA_N$, the operator $\ga_N$ that
exchanges these copies commutes with $\sA_N$ (and anticommutes with
the Dirac operator): therefore, one obtains an even spectral triple.

This construct is still metrically equivalent to the spectral triples
of subsection~\ref{ssc:Dirac-fuzzy}. Here
$\wt\sH_N \isom \sH'_N \oplus V_{N+\half}$; but the additional term
$V_{N+\half}$ carries a nontrivial subrepresentation of $\sA_N$, and
the Dirac operator $\wt D_N$ is proportional to the identity on such a
subspace: hence $[\wt D_N,a]$ vanishes on the subspace $V_{N+\half}$
for any $a \in \sA_N$. Therefore the two spectral triples induce the
same seminorm on $\sA_N$, and hence the same distance.

The authors of \cite{BP07,BP09} take another approach. Given any
finite-dimensional $\gsu(2)$-module $\Sg$, one can construct a
Dirac-like operator on $L^2(\bS^2) \ox \Sg$ by using the appropriate
representation of the abstract Dirac element \eqref{eq:D-elt}. If
$\Sg$ is the spin $j$ representation, this can be called a
``spin-$j$'' Dirac operator. For $j = \half$ we recover the ordinary
Dirac operator acting on $2$-spinors.

A spin-$\half$ Dirac operator for the fuzzy sphere is discussed in
\cite{BP07}, and is generalized to arbitrary spin~$j$ in~\cite{BP09}.
These are constructed using the Ginsparg--Wilson algebra, namely, the
free algebra generated by two grading operators $\Ga$ and~$\Ga'$. The
linear combinations $\Ga_1 = \half(\Ga + \Ga')$,
$\Ga_2 = \half(\Ga - \Ga')$, anticommute, and the proposal is to
realize them as operators on a suitable Hilbert space, interpreting
$\Ga_1$ as the Dirac operator and $\Ga_2$ as the chirality operator.
In the spin-$\half$ case, the Hilbert space is taken to be
$\sA_N \ox V_\half$. From equation (2.20) of~\cite{BP09}, or
equivalently from (8.29) of~\cite{BP07}, we see that the Dirac
operator is the same as the operator~\eqref{eq:full-Dirac} of our full
spectral triple. The chirality operator, (2.21) of~\cite{BP09}, in
contrast with the Dirac operator, is constructed using the
anticommutator with $\pi_j(J_k)$, i.e., $L_k^L + L_k^R$ in the
notation of~\cite{BP09}.

The asymmetry of the Dirac operator spectrum was already noticed
in~\cite{BP07}. At the end of subsection~8.3.2 we read:
\begin{quote}
\textsl{For} $j = 2L + 1$ [$\ell = N+1$ in our notations here]
\textsl{we get the positive eigenvalue correctly, but the negative one
is missing. That is an edge effect caused by cutting off the angular
momentum at~$2L$.}
\end{quote}
And in the same subsection, after equation (8.30):
\begin{quote}
\textsl{As mentioned earlier, use of $\Ga_2$ as chirality resolves a
difficulty addressed elsewhere~[80], where $\sign(\Ga_2)$ was used as
chirality. That necessitates projecting out $V_{+1}$ and creates a
very inelegant situation.}
\end{quote}
In other words, $\Ga_2$ is not a true grading operator. Since $\Ga_2$
anticommutes with the Dirac operator, it must vanish on $V_{N+\half}$
(otherwise, the Dirac operator would have an eigenvector $\Ga_2v$ for
$v\in V_{N+\half}$, with eigenvalue $-N-1$); which entails
$(\Ga_2)^2 \neq 1$.

A third proposal is that of~\cite{CWW97,CWW98}. It starts by 
constructing, on the Hilbert space $\sA_N$, a square~$1$ chirality operator
that is a genuine $\bZ_2$-grading, then finding a
Dirac-like operator $\nD$ by imposing anticommutation with the
grading, arriving at an even spectral triple. It follows that this
operator cannot be isospectral to our $\wt\sD_N$. The earlier paper
uses a chirality operator $\ga_\chi$, see~(5) of~\cite{CWW98}, that
does not commute with the algebra~$\sA_N$. Later, in (6)
of~\cite{CWW98}, this is corrected to $\ga_\chi^\circ$ by replacing
left with right multiplication. On imposing anticommutation of $\nD$
with that grading, one arrives at a ``second order'' operator, (8)
of~\cite{CWW98}, that in our notations is
$\nD(a \ox v) := c\, \ga_\chi^\circ \sum_{klm} \eps_{klm}\,
\pi_j(J_k)\,a\,\pi_j(J_l) \ox \sg_m v$, where $c$ is a normalization
constant.

{}From (17) of~\cite{CWW98}, relabelling with $\ell = j + \half$,
we see that the spectrum of $\nD$ is given by the eigenvalues
$\pm \la_\ell$, for $\ell = 1,\dots,N+1$, with
$$
\la_\ell^2 := \frac{\ell^2((N+1)^2 - \ell^2)}{N(N + 2)} \,.
$$
Note that $\la_\ell$ is nonlinear in $\ell$, and that $\la_{N+1} = 0$,
i.e., this operator has a kernel~$V_{N+\half}$.

The mentioned proposals, and other variants such as~\cite{HQT06}, 
begin with a chirality operator and then find an anticommuting 
self-adjoint Dirac-like operator with a plausible spectrum. Our 
approach, in contrast, starts from $SU(2)$-equivariance and arrives 
at a neater truncation of the classical spectrum, paying the price of 
spectral asymmetry.


\section{Spectral distance between coherent states} 
\label{sec:dist}

Having reduced the problem of computing distances on the fuzzy sphere,
via Prop.~\ref{pr:same-metric}, to the use of the irreducible 
spectral triple $(\sA_N,\sH_N,D_N)$, we now compute the  
distance between particular pairs of pure states in~$\sS(\sA_n)$.
Using \eqref{eq:abs-D}, we know that
\begin{equation}
D_N = \twobytwo{1 + \pi_j(H)}{\pi_j(F)}{\pi_j(E)}{1 - \pi_j(H)}
\label{eq:DN-matrix} 
\end{equation}
where again $2j = N$. From now on we omit the representation symbol
$\pi_j$ and use the matrix of~\eqref{eq:abs-D} instead, by an abuse of
notation. The spectral distance is denoted by~$d_N$.

\begin{lemma} 
\label{lm:ineq}
For any $a \in \sA_N$, the following inequalities hold:
$$
\bigl\| [H,a] \bigr\| \leq \bigl\| [D_N,a] \bigr\|,  \qquad
\bigl\| [E,a] \bigr\| \leq \bigl\| [D_N,a] \bigr\|,  \qquad
\bigl\| [F,a] \bigr\| \leq \bigl\| [D_N,a] \bigr\|.
$$
Moreover, if $a$ is a diagonal hermitian matrix, then
$\|[D_N,a]\| = \|[E,a]\|$.
\end{lemma}

\begin{proof}
Using the expression
$$
[D_N,a]^* [D_N,a] = \begin{pmatrix}
[H,a]^* [H,a] + [E,a]^* [E,a] & \cdots \\
\cdots & \cdots \end{pmatrix},
$$
we find a lower bound for $\|[D_N,a]\|$ taking the supremum over
unit vectors of the form $(x,0)^t$, with $x \in V_j$:
\begin{align*}
\bigl\| [D_N,a] \bigr\|^2 
&\geq \sup_{\|x\|=1} \bigr< x \bigm| 
([H,a]^* [H,a] + [E,a]^* [E,a] \bigr) x \bigr>
\\
&= \sup_{\|x\|=1} \bigl( \|[H,a]\,x\|^2 + \|[E,a]\,x\|^2 \bigr)
= \|[H,a]\|^2 + \|[E,a]\|^2.
\end{align*}
Thus $\|[H,a]\| \leq \bigl\| [D_N,a] \bigr\|$ and
$\|[E,a]\| \leq \bigl\| [D_N,a] \bigr\|$, Since $[F,a] = -[E,a^*]^*$,
we also get
$\|[F,a]\| \leq \|[D_N,a^*]\| = \|[D_N,a]^*\| = \|[D_N,a]\|$.

If $a \in \sA_N$ is a diagonal matrix, then $[H,a] = 0$, so that
$$
[D_N,a]^* [D_N,a] = \twobytwo{[E,a]^* [E,a]}{0}{0}{[F,a]^* [F,a]},
$$
thus $\|[D_N,a]\|$ is the greater of $\|[E,a]\|$ and $\|[F,a]\|$.
Furthermore, if $a = a^*$, then $[F,a] = -[E,a]^*$ and
$\|[E,a]\| = \|[F,a]\|$, so that $\|[D_N,a]\| = \|[E,a]\| = \|[F,a]\|$.
\end{proof}

The $SU(2)$-coherent states on $\sA_N$ were introduced 
in~\cite{ACGT72}, under the names Bloch or atomic coherent states, by
applying the rotation $R_{(\vf,\th)}$ to the ``ground'' state
$\ket{j,-j} \in V_j$. The coherent-state vectors are~\cite{ACGT72}:
\begin{equation}
\cket{\vf,\th}_N := \sum_{m=-j}^j \binom{2j}{j+m}^{\!\half} e^{-im\vf}
(\sin\tfrac{\th}{2})^{j+m} (\cos\tfrac{\th}{2})^{j-m}\, \ket{j,m}.
\label{eq:Bloch} 
\end{equation}
The corresponding vector states are denoted by
$$
\psi_{(\vf,\th)}^N(a) := \cbraCket{\vf,\th}{a}{\vf,\th}_N \,.
$$
These Bloch coherent states are for the group $SU(2)$ what the usual
harmonic oscillator coherent states are for the Heisenberg
group~\cite{Per86}. In particular, they are minimum uncertainty
states. The map $\bS^2 \to V_j$, sending the point
$(\vf,\th) \in \bS^2$ to the vector $\cket{\vf,\th}$, intertwines the
rotation action of $SU(2)$ on $\bS^2$ with the irrep $\pi_j$ on~$V_j$.
At the infinitesimal level, this is expressed by the next lemma, 
whose proof is a simple direct computation.

\begin{lemma} 
Regarding $\psi_{(\vf,\th)}^N$ as a vector state on $\sB(V_j)$, we 
find that
\begin{subequations}
\label{eq:psi} 
\begin{align}
\psi^N_{(\vf,\th)}([H,a]) &= -\ii\,\pd{}{\vf}\, \psi^N_{(\vf,\th)}(a),
\label{eq:psi-Ha} 
\\
\psi^N_{(\vf,\th)}([E,a])
&= e^{\ii\vf} \biggl( \pd{}{\th} + \ii \cot\th \pd{}{\vf} \biggr)
\psi^N_{(\vf,\th)}(a),
\label{eq:psi-Ea} 
\\
\psi^N_{(\vf,\th)}([F,a])
&= - e^{-\ii\vf} \biggl( \pd{}{\th} - \ii \cot\th \pd{}{\vf} \biggr)
\psi^N_{(\vf,\th)}(a).
\label{eq:psi-Fa} 
\end{align}
\end{subequations}
\end{lemma}


\subsection{The $N = 1$ case} 
\label{sec:djaib}

We write the general hermitian element $a = a^*\in M_2(\bC)$ as
$$
a = \twobytwo{a_0 + a_3}{a_1 + \ii a_2}{a_1 - \ii a_2}{a_0 - a_3}
= a_0\,1_2 + \avec \cdot \vec\sg,
$$
with $a_0$ real and $\avec = (a_1,a_2,a_3) \in \bR^3$. Arbitrary (not
necessarily pure) states on $M_2(\bC)$ are given by
$\om_\xvec(a) := a_0 + \xvec \cdot \avec$, with $\xvec$ in the closed
unit ball $B^3 \subset \bR^3$. This state is pure if and only if
$\xvec = (\sin\th \cos\vf, \sin\th \sin\vf, \cos\th)$ lies on the
boundary $\bS^2$ of the ball, in which case it coincides with the
coherent state $\psi^1_{(\vf,\th)}$. Note that \textit{for $N = 1$,
all pure states are coherent states}.

The next proposition shows that the distance among states is half of
the Euclidean distance in the ball; thus, for coherent states,
\textit{half of the chordal distance} on the sphere.

\begin{prop} 
\label{pr:djaib}
For all $\xvec,\yvec \in B^3$, the distance between the corresponding
states is
\begin{equation}
\label{eq:djaib} 
d_1(\om_\xvec, \om_\yvec) = \frac{1}{2}\, |\xvec - \yvec|.
\end{equation}
In particular, $d_1(\psi^1_{(0,\th)},\psi^1_{(0,0)}) = \sin(\th/2)$.
\end{prop}

\begin{proof}
Writing $a_\pm = a_1 \pm \ii a_2$ and $\sg_\pm = \sg_1 \pm \ii\sg_2$, 
we get, for $a = a^*$:
$$
[D_1,a]
= \twobytwo{\half[\sg_3,a]}{[\sg_-,a]}{[\sg_+,a]}{-\half[\sg_3,a]}
= \begin{pmatrix}
0 & a_+ & -a_+ & 0 \\
-a_- & 0 & 2a_3 & a_+ \\
a_- & -2a_3 & 0 & -a_+ \\
0 & -a_- & a_- & 0  \end{pmatrix}.
$$
The matrix $\ii[D_1,a]$ is hermitian, and its characteristic
polynomial is easily seen to be
$\det(\la - \ii[D_1,a]) = \la^2(\la^2 - 4|\avec|^2)$, showing that
its norm is $\|[D_1,a]\| = 2|\avec|$.

The Cauchy--Schwarz inequality
$$
\bigl| \om_\xvec(a) - \om_\yvec(a) \bigr|
= \bigl| (\xvec - \yvec) \cdot \avec \bigr| \leq |\xvec - \yvec|\,|\avec|
$$
is saturated when $\avec$ is parallel to $\xvec - \yvec$. Thus
$d_1(\om_\xvec, \om_\yvec)$ is the supremum of
$|\xvec - \yvec|\,|\avec|$ over hermitian $a$ with
$\|[D_1,a]\| = 2|\avec| \leq 1$. This establishes~\eqref{eq:djaib}.

If $\xvec = (\sin\th, 0, \cos\th)$ and $\yvec = (0,0,1)$, then
$|\xvec - \yvec|^2 = 2(1 - \cos\th) = 4\,\sin^2(\th/2)$, and thus
$d_1(\om_\xvec, \om_\yvec) = \sin(\th/2)$.
\end{proof}


\subsection{Distances between basis vectors} 
\label{ssc:basis}

Similarly to Prop.\ 3.6 of \cite{CDMW11}, the distance between basis
vectors can be exactly computed. For fixed $N = 2j$, and
$m \in \{-j,\dots,j\}$, the basic vector states are
$$
\om_m(a)  :=  \braCket{j,m}{a}{j,m} \,.
$$

\begin{prop} 
\label{pr:basis}
For any $m < n$ in $\{-j,\dots,j\}$, the following distance formula 
holds: 
\begin{equation}
\label{eq:basis} 
d_N(\om_m,\om_n) = \sum_{k=m+1}^n \frac{1}{\sqrt{(j+k)(j-k+1)}}\,.
\end{equation}
\end{prop}

\begin{proof}
If $a \in \sA_N$, then
\begin{align*}
\om_m(a) - \om_n(a)
&= \sum_{k=m+1}^n \braCket{j,k-1}{a}{j,k-1} - \braCket{j,k}{a}{j,k}
\\
&= \sum_{k=m+1}^n \frac{1}{\sqrt{(j + k)(j - k + 1)}}\,
\braCket{j,k}{[E,a]}{j,k-1}.
\end{align*}
Using Lemma~\ref{lm:ineq}, we get the estimate
$$
\bigl| \braCket{j,k}{[E,a]}{j,k-1} \bigr| \leq \|[E,a]\|
\leq \|[D_N,a]\|
$$
which shows that the left hand side of~\eqref{eq:basis} is no greater
than the right hand side. On the other hand, let $\hat a$ be the
self-adjoint diagonal operator:
\begin{equation}
\label{eq:hata} 
\hat a\,\ket{j,m} := - \biggl( \sum_{k=-j+1}^m
\frac{1}{\sqrt{(j + k)(j - k + 1)}} \biggr) \ket{j,m}.
\end{equation}
The coefficients are chosen so that
$[E,\hat a]\,\ket{j,m} = \ket{j,m+1}$ for $m = -j,\dots, j-1$. Notice 
that $\hat a\,\ket{j,-j} = 0$ and $[E,\hat a]\,\ket{j,-j} = 0$. Since
$\hat a = \hat a^*$, Lemma~\ref{lm:ineq} then shows that
$\|[D_N,\hat a]\| = \|[E,\hat a]\| = 1$. Therefore,
$$
d_N(\om_m,\om_n) \geq \om_m(\hat a) - \om_n(\hat a)
= \sum_{k=m+1}^n \frac{1}{\sqrt{(j + k)(j - k + 1)}} \,.
\eqno \qed
$$
\hideqed
\end{proof}

Note that the distance is additive on the chain of basic vector
states: $d_N(\om_m,\om_n) = \sum_{k=m+1}^n d_N(\om_{k-1}, \om_k)$.

\begin{corl} 
\label{cr:north-south}
For any $N$, the distance between the north and south poles of the
fuzzy sphere is:
\begin{equation}
\label{eq:diameter} 
d_N(\psi^N_{(0,0)}, \psi^N_{(0,\pi)})
= \sum_{k=1}^N \frac{1}{\sqrt{k(N - k + 1)}} \,.
\end{equation}
\end{corl}

\begin{proof}
By construction, the Bloch state vectors at the poles are basis 
vectors: $\cket{0,0}_N = \ket{j,-j}$ and $\cket{0,\pi}_N = \ket{j,j}$.
Therefore, $\psi^N_{(0,0)} = \om_{-j}$ and $\psi^N_{(0,\pi)} = \om_j$.
{}From~\eqref{eq:basis} we get~\eqref{eq:diameter}, since the left 
hand side is just $d_N(\om_{-j}, \om_j)$.
\end{proof}


\subsection{An auxiliary distance} 
\label{ssc:aux-dist}

Let $\sB_N \subset \sA_N$ be the subalgebra of \textit{diagonal}
matrices. Note that if $a$ is diagonal, then
$\psi_{(\vf,\th)}^N(a) = \psi_{(0,\th)}^N(a)$ for any~$\vf$. Define the 
distance
\begin{equation}
\label{eq:aux-dist} 
\rho_N(\th) := \sup \bigl\{ 
\bigl| \psi_{(0,\th)}^N(a) - \psi_{(0,0)}^N(a) \bigr|
: a =a ^* \in \sB_N, \ \|[D_N,a]\| \leq 1 \bigr\}.
\end{equation}

\begin{prop} 
\label{pr:rho-N}
For any $0 \leq \th \leq \pi$, $\rho_N(\th)$ is given by:
\begin{equation}
\label{eq:rho-N} 
\rho_N(\th) = \sum_{n=1}^N \binom{N}{n}
(\sin\tfrac{\th}{2})^{2n} (\cos\tfrac{\th}{2})^{2(N-n)}
\sum_{k=1}^n \frac{1}{\sqrt{k(N - k + 1)}}\,.
\end{equation}
\end{prop}

\begin{proof}
Let $a = (\delta_{mn} c_m) \in \sB_N$, with $c_m \in \bR$. Then
$\om_m(a) = c_m$, which gives
$$
\psi^N_{(0,0)}(a) - \psi^N_{(0,\th)}(a)
= \sum_{m=-j}^j \binom{2j}{j+m} (\sin\tfrac{\th}{2})^{2(j+m)}
(\cos\tfrac{\th}{2})^{2(j-m)} (\om_{-j}(a) - \om_m(a)).
$$
We also know that
$$
\om_{-j}(a) - \om_m(a) \leq d_N(\om_m, \om_{-j})
= \sum_{m'=-j+1}^m \frac{1}{\sqrt{(j + m')(j - m' + 1)}}
$$
for all $a$ with $\|[D_N,a]\| \leq 1$, with the supremum saturated on
the diagonal element $\hat a$ given by~\eqref{eq:hata}. On
substituting $n = j + m$ and $k = j + m'$, we arrive 
at~\eqref{eq:rho-N}.
\end{proof}

\begin{lemma} 
\label{lm:deriv}
The derivative $\rho'_N(\th)$ of \eqref{eq:rho-N} satisfies
$0 \leq \rho'_N(\th) \leq 1$.
\end{lemma}

\begin{proof}
{}From \eqref{eq:psi-Ea} we deduce that
$\psi^N_{(0,\th)}([E,a]) = \pd{}{\th} \psi^N_{(0,\th)}(a)$ for all
$a \in \sB_N$. Using this relation and the equality
$\rho_N(\th) = \psi^N_{(0,\th)}(\hat a) - \psi^N_{(0,0)}(\hat a)$,
with $\hat a$ the element in~\eqref{eq:hata}, we get:
$$
\rho'_N(\th) = \pd{}{\th} \psi^N_{(0,\th)}(\hat a)
= \psi^N_{(0,\th)}([E, \hat a]).
$$
Since states are functionals with norm~$1$, it follows that
$$
|\rho'_N(\th)| = |\psi^N_{(0,\th)}([E,\hat{a}])| \leq  
\psi^N_{(0,\th)}(1)\, \|[E,\hat{a}]\| = 1.
$$
On the other hand, since $L := [E, \hat a]$ is the ladder operator
$\ket{j,m} \mapsto \ket{j,m+1}$, we get
$$
\rho'_N(\th) = \cbraCket{0,\th}{L}{0,\th}_N
= \sum_{m=-j}^{j-1} \binom{2j}{j+m}^{\!\half}
\binom{2j}{j+m+1}^{\!\half} (\sin\tfrac{\th}{2})^{2j+2m+1}
(\cos\tfrac{\th}{2})^{2j-2m-1} \geq 0.
$$
Actually, we see that $\rho'_N(\th) > 0$ for $0 < \th < \pi$.
\end{proof}

The previous lemma has two consequences: $\rho_N(\th)$ 
is strictly increasing on $0 \leq \th \leq \pi$, for fixed~$N$; and, 
for $0 < \th \leq \pi$ the mean value theorem gives $\phi$ 
with $0 < \phi < \th$ such that 
$$
\rho_N(\th) = \rho_N(\th) - \rho_N(0) = \th\,\rho'_N(\phi) \leq \th.
$$
That is: $\rho_N(\th)$ is no greater than the geodesic distance on the
circle.


\subsection{$SU(2)$-invariance of the distance} 
\label{ssc:d-round}

\begin{lemma} 
\label{lm:d-round}
The distance function $d_N(\psi^N_{(\vf,\th)},\psi^N_{(\vf',\th')})$
is $SU(2)$-invariant.
\end{lemma}

\begin{proof}
Up to now, we have identified the element $a\in \sA_N \isom \End(V_j)$
with the operator $a \ox 1_2$ acting on $\sH_N = V_j \ox V_\half$. In
this proof, we shall write explicitly $a \ox 1_2$ to avoid
ambiguities.

For any $g \in SU(2)$ and $a \in \sA_N$, we write
$a^g := \pi_j(g) a \pi_j(g)^*$. Since
$\pi_\half(g) \pi_\half(g)^* = 1_2$ by unitarity of $\pi_\half$, we 
get
$$
a^g \ox 1_2 = u(a \ox 1_2)u^*
\word{where} u := \pi_j(g) \ox \pi_\half(g).
$$
Since $D_N$ commutes with $u$, the operator
$[D_N, a^g \ox 1_2] = u[D_N, a \ox 1_2]u^*$ has the same norm as
$[D_N, a \ox 1_2]$.

Given a state $\om$ on~$\sA_N$ and $g \in SU(2)$, let $g_*\om$ be the
state defined by $g_*\om(a) = \om(a^g)$. For any pair of states
$\om,\om'$, we then obtain
\begin{align*}
d_N(g_*\om,g_*\om') &= \sup_{a\in\sA_N}
\set{|\om(a^g) - \om'(a^g)| : \|[D_N, a \ox 1_2]\| \leq 1}
\\
&= \sup_{b\in\sA_N}
\set{|\om(b) - \om'(b)| : \|[D_N, b \ox 1_2]\| \leq 1} = d_N(\om,\om'),
\end{align*}
where we have put $b = a^g$ and used
$\|[D_N, a^g \ox 1_2]\| = \|[D_N, a \ox 1_2]\|$. By construction, the
action $\psi^N_{(\vf,\th)} \mapsto g_*\psi^N_{(\vf,\th)}$ corresponds
to the usual rotation action of $SU(2)$ on $\bS^2$.
\end{proof}


\subsection{Dependence on the dimension} 
\label{ssc:N-depend}

We now show that the distance
$d_N(\psi^N_{(\vf,\th)},\psi^N_{(\vf',\th')})$ is non-decreasing
with~$N$. Using the fuzzy spinor basis~\eqref{eq:spinors}, one defines
injections $U_j^\pm \: V_{j\pm\half} \to V_j \ox V_\half$ by
$$
U_j^+ \ket{j + \half, m + \half} := \kket{j,m}_+ \,,  \qquad
U_j^- \ket{j - \half, m + \half} := \kket{j,m}_- \,,
$$
using the same index sets as in~\eqref{eq:spinors}, namely
$m = -j-1, \dots, j$ for the range of~$U_j^+$ and $m = -j, \dots, j-1$
for the range of~$U_j^-$. One easily checks that these $U_j^\pm$ are
isometries, i.e., $(U^\pm_j)^*U^\pm_j = 1$, that intertwine the
representations of~$\gsu(2)$. Also, $V_j \ox V_\half$ is the
orthogonal direct sum of the ranges of $U_j^+$ and~$U_j^-$.

\begin{lemma} 
\label{lm:cohere}
$U^+_j\cket{\vf,\th}_{N+1}  = \cket{\vf,\th}_N \ox \cket{\vf,\th}_1$
for any $(\vf,\th) \in \bS^2$.
\end{lemma}

\begin{proof}
Note that $\cket{\vf,\th}_1
= e^{-\half\ii\vf} \sin \frac{\th}{2}\, \ket{\half,-\half}
+ e^{\half\ii\vf} \cos \frac{\th}{2}\, \ket{\half,\half}$. The rest is
an easy computation, using \eqref{eq:Bloch}.
\end{proof}

We define two injective linear maps
$$
\eta_N^\pm : \sA_N \to \sA_{N\pm 1} \,,  \qquad
\eta_N^\pm(a) := (U_j^\pm)^* (a \ox 1_2) U_j^\pm \,.
$$
They are unital and commute with the involution, but are neither
surjective nor algebra morphisms, since 
$U_j^+ (U_j^+)^* + U_j^- (U_j^-)^* = 1$. They are norm decreasing: the
norm of $a \ox 1_2$ on the range of $U^\pm_j$ is no greater than its
norm on $V_j \ox V_\half$, which equals the norm of $a$ on~$V_j$.

\begin{lemma} 
\label{lm:step}
For any $a \in \sA_N$,
\begin{equation}
\label{eq:step-psi} 
\psi^{N+1}_{(\vf,\th)} \circ \eta_N^+(a) = \psi^N_{(\vf,\th)}(a),
\end{equation}
and
\begin{equation}
\label{eq:step-D} 
\bigl\| [D_{N\pm 1}, \eta_N^\pm(a)] \bigr\|
\leq \bigl\| [D_N, a] \bigr\|.
\end{equation}
\end{lemma}

\begin{proof}
The equality~\eqref{eq:step-psi} follows from Lemma \ref{lm:cohere},
because
$$
\cbraCket{\vf,\th}{\eta_N^+(a)}{\vf,\th}_{N+1}
= \cbraCket{\vf,\th}{a}{\vf,\th}_N \, \cbraket{\vf,\th}{\vf,\th}_1
= \cbraCket{\vf,\th}{a}{\vf,\th}_N \,.
$$

Since $U_j^\pm$ intertwines representations of $\gsu(2)$, i.e.,
$$
U_j^\pm X = (X \ox 1_2 + 1_2 \ox X) U_j^\pm
\word{for all} X \in \gsu(2)
$$
(the representation symbols are omitted), we conclude that
$$
[X, \eta_N^\pm(a)] = (U_j^\pm)^* \bigl( [X,a] \ox 1_2 \bigr) U_j^\pm
= \eta_N^\pm\bigl( [X,a] \bigr).
$$
In view of~\eqref{eq:DN-matrix}, therefore,
$[D_{N\pm 1}, \eta_N^\pm(a)] = \eta_N^\pm \bigl( [D_N, a] \bigr)$,
where $[D_N,a] \in M_2(\sA_N)$ and we extend $\eta_N^\pm$ from $\sA_N$
to $M_2(\sA_N)$ by applying it to each matrix entry. Since both
$\eta_N^\pm$ are norm-decreasing maps, this proves~\eqref{eq:step-D}.
\end{proof}

\begin{prop} 
\label{pr:dN-grows}
For any $N \geq 1$, the following majorization holds:
$$
d_{N+1}(\psi^{N+1}_{(\vf,\th)}, \psi^{N+1}_{(\vf',\th')})
\geq d_N(\psi^N_{(\vf,\th)}, \psi^N_{(\vf',\th')}).
$$
\end{prop}

\begin{proof}
We get directly:
\begin{align*}
d_{N+1}(\psi^{N+1}_{(\vf,\th)}, \psi^{N+1}_{(\vf',\th')})
&= \sup_{a\in\sA_{N+1}} \bigl\{\, \bigl| \psi^{N+1}_{(\vf,\th)}(a) 
- \psi^{N+1}_{(\vf',\th')}(a) \bigr| 
: \|[D_{N+1}, a]\| \leq 1 \,\bigr\}
\\
&\geq \sup_{a\in\sA_N}
\bigl\{\, \bigl| \psi^{N+1}_{(\vf,\th)} \circ \eta_N^+(a)
- \psi^{N+1}_{(\vf',\th')} \circ \eta_N^+(a) \bigr|
: \bigl\| [D_{N+1},\eta_N^+(a)] \bigr\| \leq 1 \,\bigr\}
\\
&= \sup_{a\in\sA_N}
\bigl\{\,\bigl| \psi^N_{(\vf,\th)}(a) - \psi^N_{(\vf',\th')}(a) \bigr|
: \bigl\| [D_{N+1}, \eta_N^+(a)] \bigr\| \leq 1 \,\bigr\}
\\
&\geq \sup_{a\in\sA_N}
\bigl\{\,\bigl| \psi^N_{(\vf,\th)}(a) - \psi^N_{(\vf',\th')}(a) \bigr|
: \|[D_N,a]\| \leq 1 \,\bigr\}
\\
&=d_N(\psi^N_{(\vf,\th)},\psi^N_{(\vf',\th')}) \;.
\end{align*}
The first inequality follows since the supremum over the range of
$\eta_N^+$ in $\sA_{N+1}$ is smaller than the supremum over the whole
$\sA_{N+1}$. In the next line \eqref{eq:step-psi} is used; and we get
the second inequality from~\eqref{eq:step-D}.
\end{proof}

\begin{remk} 
\label{rk:rhoN-grows}
The calculation in the proof of Prop.~\ref{pr:dN-grows} can be 
adapted to establish that
\begin{equation}
\label{eq:rhoN-grows} 
\rho_{N+1}(\th - \th') \geq \rho_N(\th - \th'),
\word{for}  \th,\th' \in [0,\pi].
\end{equation}
For that, just restrict $a \in \sA_N$ to (be self-adjoint and) lie in 
the diagonal subalgebra $\sB_N$. The only thing to note that is that 
$\eta_N^+$ maps $\sB_N$ into a non-diagonal subalgebra of $\sA_{N+1}$;
but the notion of diagonal subalgebra is in any case basis-dependent. 
It is enough to replace $\sB_{N+1}$ by a conjugate subalgebra that
includes $\eta_N^+(\sB_N)$, after conjugating $\sA_{N+1}$ by a 
unitary operator commuting with the $SU(2)$ action via 
$\ad\pi_{j+\half}$. This rotates the basis vectors in $V_{j+\half}$,
in such a way that the coherent states $\psi^{N+1}_{(\vf,\th)}$ are 
unchanged. Thus also, $\rho_{N+1}(\th - \th')$ is unchanged, and 
\eqref{eq:rhoN-grows} holds.
\end{remk}


\subsection{Upper and lower bounds and the large $N$ limit} 
\label{ssc:bounds}

\begin{prop} 
\label{pr:bounds}
The following inequalities hold, for all
$(\vf,\th), (\vf',\th')\in \bS^2$:
\begin{equation}
\label{eq:sandwich} 
\rho_N(\th - \th') \leq d_N(\psi^N_{(\vf,\th)}, \psi^N_{(\vf',\th')})
\leq d_\geo \bigl( (\vf,\th),(\vf',\th') \bigr),
\end{equation}
where $\rho_N(\th)$ is the auxiliary distance~\eqref{eq:aux-dist} and
$d_\geo$ is the geodesic distance for the round metric of~$\bS^2$. In
particular,
\begin{equation}
\label{eq:bounds} 
\rho_N(\th) \leq d_N(\psi^N_{(0,\th)}, \psi^N_{(0,0)}) \leq \th.
\end{equation}
\end{prop}

\begin{proof}
Due to Lemma~\ref{lm:d-round}, the second inequality
in~\eqref{eq:sandwich} involves two $SU(2)$-invariant expressions. It
is then enough to prove it when $(\vf',\th') = (0,\pihalf)$ and
$(\vf,\th) = (\vf,\pihalf)$. We thus need to prove that
$$
d_N(\psi^N_{(\vf,\pihalf)}, \psi^N_{(0,\pihalf)}) \leq  |\vf| 
\word{for all}  -\pi < \vf \leq \pi.
$$
Integrating \eqref{eq:psi-Ha}, we find
$$
\psi^N_{(\vf,\pihalf)}(a) - \psi^N_{(0,\pihalf)}(a)
= \ii \int_0^\vf \psi^N_{(\al,\pihalf)}([H,a]) \,d\al \,,
$$
and since $|\om(A)| \leq \|A\|$ for any state $\om$ and operator~$A$,
we obtain, using Lemma~\ref{lm:ineq}:
$$
|\psi^N_{(\vf,\pihalf)}(a) - \psi^N_{(0,\pihalf)}(a)|
\leq \|[H,a]\|\, \biggl| \int_0^\vf d\al \biggr| = |\vf|\, \|[H,a]\|
\leq |\vf|\, \|[D_N,a]\|.
$$
This proves the upper bound in~\eqref{eq:sandwich}. That
of~\eqref{eq:bounds} follows from
$d_\geo\bigl((0,\th),(0,0)\bigr) = \th$.

A lower bound for the distance is given by the supremum over diagonal
matrices:
$$
d_N(\psi^N_{(\vf,\th)}, \psi^N_{(\vf',\th')})
\geq \sup_{a=a^*\in\sB_N}
\bigl\{\,\bigl| \psi^N_{(\vf,\th)}(a) - \psi^N_{(\vf',\th')}(a) \bigr|
: \|[D_N,a]\| \leq 1 \,\bigr\}.
$$
Since $\psi^N_{(\vf,\th)}(a)$ is independent of~$\vf$ 
for any diagonal~$a$, we arrive at
$$
d_N(\psi^N_{(\vf,\th)}, \psi^N_{(\vf',\th')})
\geq \sup_{a=a^*\in\sB_N}
\bigl\{\, \bigl| \psi^N_{(0,\th)}(a) - \psi^N_{(0,\th')}(a) \bigr|
: \|[D_N,a]\| \leq 1 \,\bigr\}
= \rho_N(\th - \th').
\eqno \qed
$$
\hideqed
\end{proof}

For $0 < \th < \pi$, neither the upper nor the lower bound
in~\eqref{eq:bounds} is sharp. On the other hand,
$d_N(\psi^N_{(0,\pi)}, \psi^N_{(0,0)}) = \rho_N(\pi)$, since the
formula \eqref{eq:rho-N} coincides with \eqref{eq:diameter} when
$\th = \pi$. Thus the lower bound is sharp for $\th = \pi$. In
Figure~\ref{fig:rho-asymp} we show a plot of the upper bound (straight
line) and lower bounds $\rho_N$ for $N = 10,30,500$ (nondecreasing
with~$N$). It would seem that $\th - \rho_N(\th)$ has its maximum at
$\th = \pi$. Figure~\ref{fig:rho-drop} plots $\th - \rho_N(\th)$ for
$N = 5,10,20,30$ (decreasing with~$N$). This suggests how to prove our
final result.

\begin{figure}[htb] 
\centering
\subfloat[Plot of $\th$ and $\rho_N(\th)$ for $N = 10,30,500$.]{
\includegraphics[width=6.5cm]{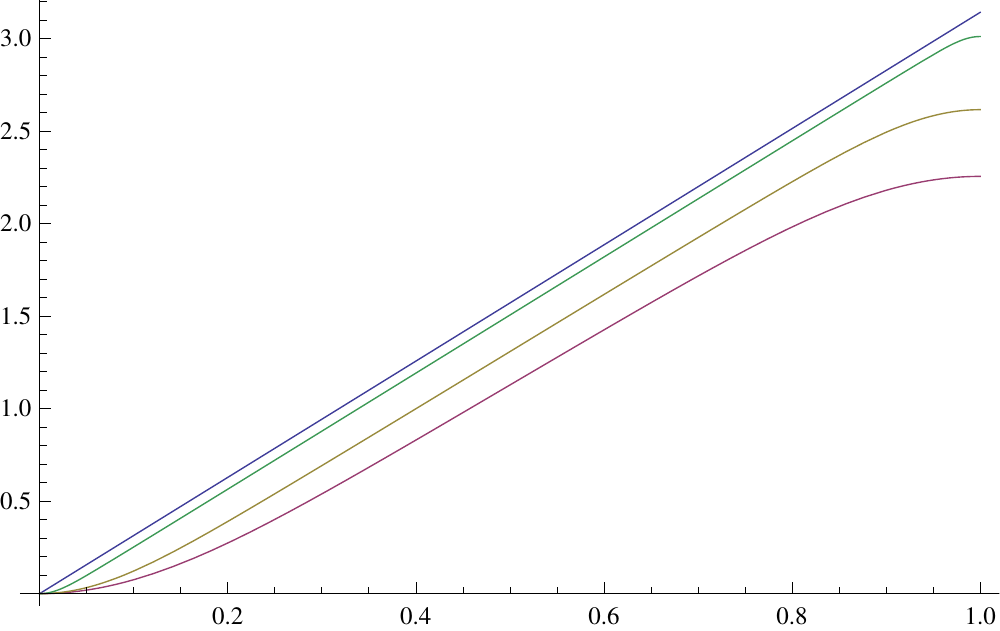}\label{fig:rho-asymp}}
\hspace{1cm}
\subfloat[Plot of $\th - \rho_N(\th)$ for $N = 5,10,20,30$.]{
\includegraphics[width=6.5cm]{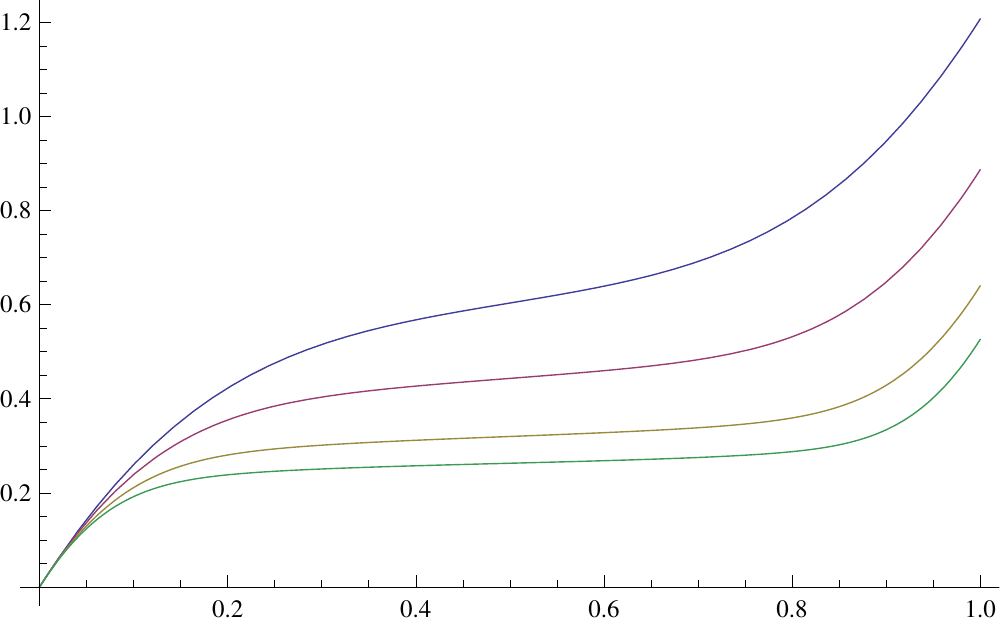}\label{fig:rho-drop}}
\caption{Comparison of $\rho_N(\th)$ with~$\th$. The abscissa is
$x = \pi\th$.}
\end{figure}

\begin{prop} 
\label{pr:finale}
As $N \to \infty$, the sequence $\rho_N(\th)$ is uniformly convergent
to $\th$ in $[0,\pi]$. Therefore,
$$
\lim_{N\to\infty} d_N(\psi^N_{(\vf,\th)}, \psi^N_{(\vf',\th')})
= d_\geo \bigl( (\vf,\th), (\vf',\th') \bigr).
$$
\end{prop}

\begin{proof}
Let $f_N(\th) := \th - \rho_N(\th)$. Clearly $f_N(0) = 0$, and
$f'_N(\th) \geq 0$ by Lemma~\ref{lm:deriv}. Hence $f_N(\th)$ is a 
nondecreasing positive function of~$\th$ for each~$N$, and
$$
\|\th - \rho_N(\th)\|_\infty = \sup_{\th\in[0,\pi]} f_N(\th)
\leq f_N(\pi) = \pi - \rho_N(\pi).
$$
Therefore, the uniform convergence
$\lim_{N\to\infty} \|\th - \rho_N(\th)\|_\infty = 0$ holds if and only
if the diameter converges to~$\pi$, i.e., 
$\lim_{N\to\infty} \rho_N(\pi) = \pi$.

The formula for $\rho_N(\pi)$ is given by \eqref{eq:diameter}. The
sequence $\rho_N(\pi)$ is bounded, $\rho_N(\pi) \leq \pi$, and is
nondecreasing by Remark~\ref{rk:rhoN-grows}. Hence it is convergent,
and the limit can be computed using any subsequence. It is thus enough
to prove that $\rho_N(\pi) \geq c_N$, where $c_N \to \pi$ as
$N \to \infty$.

We consider the subsequence with odd $N$ only. The function
$(x(N - x + 1))^{-1/2}$ is positive for $1 \leq x \leq N$, symmetric
about $x = \half(N + 1)$, and monotonically decreasing for 
$1 \leq x \leq \half(N + 1)$. Hence
$$
\rho_N(\pi) = 2 \sum_{k=1}^{\half(N-1)} \frac{1}{\sqrt{k(N - k + 1)}}
+ \frac{2}{N + 1}
\geq 2 \int_1^{\half(N+1)} \frac{dx}{\sqrt{x(N - x + 1)}} \,.
$$
Substituting $x =: \half(N + 1)(1 + \sin\xi)$, so that
$d\xi = dx/\sqrt{x(N - x + 1)}\,$, we obtain
$$
\rho_N(\pi) \geq 2 \arcsin \frac{N - 1}{N + 1} \,.
$$
The right hand side converges monotonically to~$\pi$ as $N\to\infty$,
thus $\lim_{N\to\infty} \rho_N(\pi) = \pi$ through odd~$N$, and so, as
noted above, through all~$N$. (A slightly modified estimate gives
$\lim_{N\to\infty} \rho_N(\pi) = \pi$ through even~$N$, directly,
without using Remark~\ref{rk:rhoN-grows}.) This proves the uniform
convergence $\rho_N(\th) \to \th$.

The estimate~\eqref{eq:bounds} now shows that
$d_N(\psi^N_{(0,\th)},\psi^N_{(0,0)})$ is uniformly convergent
to~$\th$; by $SU(2)$-invariance, 
$d_N(\psi^N_{(\vf,\th)},\psi^N_{(\vf',\th')})$ converges to
$d_\geo \bigl((\vf,\th), (\vf',\th')\bigr)$ uniformly on~$\bS^2$.
\end{proof}


\paragraph{Acknowledgments}

FL acknowledges support by CUR Generalitat de Catalunya under project
FPA2010--20807. FD and FL acknowledge the \textit{Faro} project
\textit{Algebre di Hopf}
of the Universit\`a di Napoli \textit{Federico II}. JCV thanks the
Universit\`a di Napoli for warm hospitality and acknowledges support
from the Vicerrector\'ia de Investigaci\'on of the Universidad de
Costa~Rica.


\end{document}